\documentclass[submission,copyright,creativecommons]{eptcs}
\providecommand{\event}{LINEARITY 2012} 
\usepackage{stmaryrd}
\usepackage{amsmath}
\usepackage{amsthm}
\usepackage{amssymb}
\usepackage{amscd}
\usepackage{proof}
\usepackage{cmll} 
\usepackage{float}
\floatstyle{boxed} 
\restylefloat{figure}
\usepackage{xcolor} 
\usepackage{xkeyval}
\usepackage{tikz}
\usepackage{enumerate}

\usepackage{subfigure}

\usepackage{aliascnt} 
\def\mynewtheorem#1[#2]#3{%
  \newaliascnt{#1}{#2}%
  \newtheorem{#1}[#1]{#3}%
  \aliascntresetthe{#1}%
  \expandafter\def\csname #1autorefname\endcsname{#3}%
}
\usepackage{amsthm}

\newtheorem{theorem}{Theorem}

\mynewtheorem{lemma}[theorem]{Lemma}
\mynewtheorem{corollary}[theorem]{Corollary}
\mynewtheorem{proposition}[theorem]{Proposition}
\theoremstyle{definition}
\mynewtheorem{remark}[theorem]{Remark}
\mynewtheorem{convention}[theorem]{Convention}
\mynewtheorem{fact}[theorem]{Fact}
\mynewtheorem{example}[theorem]{Example}
\mynewtheorem{definition}[theorem]{Definition}
\mynewtheorem{property}[theorem]{Property}
\mynewtheorem{notation}[theorem]{Notation}
\usepackage{array}
\newcolumntype{b}{@{}>{{}}} 
\newcolumntype{B}{@{}>{{}}c<{{}}@{}}
\newcolumntype{h}[1]{@{\hspace{#1}}}
\newcolumntype{L}{>{$}l<{$}}
\newcolumntype{C}{>{$}c<{$}}
\newcolumntype{R}{>{$}r<{$}}
\newcolumntype{S}{>{$(}r<{)$}}
\newcolumntype{n}{@{}}

\newenvironment{enumcases}
  {\begin{enumerate}[\bfseries{Case} 1.]}
  {\end{enumerate}}

\newcommand{\bag}[1]{[#1]}
\newcommand{\setof}[1]{\{\,#1\,\}}
\newcommand{\bang}[1]{{#1}^{\oc}}
\newcommand{\argu}[1]{{#1}^{(\oc)}}
\newcommand{\app}[2]{#1\!\cdot\! #2}
\newcommand{\linear}[2]{\langle#1/#2\rangle}  
\newcommand{\expo}[2]{\left\{#1/#2\right\}}  
\newcommand{\gen}[2]{\langle\!\langle#1/#2\rangle\!\rangle}  
\newcommand{\sdot}{\!\cdot\!} 
\newcommand\Lamr{\ensuremath{\Lambda^{r}}} 
\newcommand\Expr{\Lambda^{(b)}} 
\newcommand{\Nat}{\mathtt{Nat}} 
\newcommand{\Sum}[1]{\mathbb{#1}} 

\newcommand{\dpap}{\llbracket}
\newcommand{\dpch}{\rrbracket}
\newcommand{\hole}[1]{\dpap #1\dpch} 
 
\def\maysol{may-solvable}

\DeclareMathOperator{\hnf}{onf} 
\DeclareMathOperator\mhnf{monf} 

\newcommand*{\mystackrel}[2]{\stackrel{\raise-2pt\hbox{$\scriptstyle\!#1$}}{#2}}
\newcommand*{\mystackrelrev}[2]{\stackrel{\raise-2pt\hbox{$\scriptstyle\,#1$}}{#2}}
\newcommand*{\mystackrelequi}[2]{\stackrel{\raise-2pt\hbox{$\scriptstyle#1$}}{#2}}

\newcommand*{\red}[1]{\mathtt{#1}}  
\newcommand*{\redto}[1][]{\mystackrel{\red{#1}}{\rightarrow}}
\newcommand*{\xredto}[1][]{\xrightarrow{\raisebox{-2pt}{$\scriptstyle\!\red{#1}$}}}

\newcommand{\outbeta}[1][]{\xredto[\mathtt{o} #1]} 
\newcommand{\nonoutbeta}[1][]{\xredto[\mathtt{i} #1]} 

\renewcommand{\L}{\mathcal L}

\newcommand{\bnf}{\mathrel{::=}}
\newcommand{\ass}{\mathrel{:=}}

\newcommand{\textdef}[1]{\textbf{#1}}
\newcommand{\lam}{\ensuremath{\lambda}}

\newcommand{\NDMachine}{\ensuremath{\Downarrow_{nd}}} 
\newcommand{\DNDMachine}{\ensuremath{\Uparrow_{nd}}} 
\newcommand{\BagMachine}{\ensuremath{\Downarrow_b}}

\allowdisplaybreaks

\usepackage[colorinlistoftodos,textsize=footnotesize,textwidth=2.4cm,%
disable 
]{todonotes}
\title{Standardization in resource lambda-calculus}
\author{Maurizio Dominici \;\;\; Simona Ronchi Della Rocca
\institute{Dipartimento di Informatica -- Universit\`a di Torino }
\email{dominicimaurizio@gmail.com \;\;\; ronchi@di.unito.it }\\
Paolo Tranquilli
\institute{Dipartimento di Scienze dell'Informazione -- Universit\`a di Bologna }
\email{tranquil@cs.unibo.it}
}
\def\titlerunning{Standardization in resource lambda-calculus}
\def\authorrunning{M. Dominici \& S. Ronchi Della Rocca \& P. Tranquilli}

\begin{document}
\maketitle
\begin{abstract}
The resource calculus is an extension of the $\lambda$-calculus allowing to model resource consumption.
It is intrinsically non-deterministic and has two general notions of reduction -- one parallel, preserving all the possible
results as a formal sum, and one non-deterministic, performing an exclusive choice at every step.
We prove that the non-deterministic reduction enjoys a notion of standardization, which is the natural extension
with respect to the similar one in classical $\lambda$-calculus. The full parallel reduction only enjoys a weaker notion
of standardization instead.
The result allows an operational characterization of may-solvability, which has been introduced and already characterized (from the
syntactical and logical points of view) by Pagani and 
Ronchi Della Rocca.
\end{abstract}
\section{Introduction}
The resource calculus ($\Lamr$) is an extension of the $\lambda$-calculus allowing to model resource consumption. Namely, 
the argument of a function comes as a finite multiset of resources, which in turn can be either linear or reusable. A linear
resource must be used exactly once, while a reusable one can be called \emph{ad libitum}. In this setting the 
evaluation of a function applied to a multiset of resources gives rise to different possible choices, because of the different
possibilities of distributing the resources among the occurrences of the formal parameter. 
We can define two kinds of reduction, according to the interpretation we want to give to this fact. 
The parallel reduction (which can be further divided in giant and baby) performs all the possible choices, and gives as result a formal sum preserving all the possible results, while
the non-deterministic reduction at every step chooses non-deterministically one of the possible results. 
In case of a multiset of linear resources, also a notion of \emph{crash} arises, whenever the cardinality 
of the multiset does not fit exactly the number of occurrences. 
Then the resource calculus is a useful framework for studying the notions of linearity and non-determinism, and the 
relation between them.
$\Lamr$ is a descendant of the calculus of multiplicities, introduced by Boudol in \cite{boudrescalc}, and it has been designed 
by
Tranquilli \cite{intudiffnet} in order to give a precise syntax for the  differential $\lambda$-calculus of Ehrhard and Regnier 
\cite{lambdadiff}.
$\Lambda^{r}$ can be used as a paradigmatic language for different kinds of computation.
Usual $\lambda$-calculus can be embedded in it. Forbidding linear terms but allowing non-empty finite multisets of reusable 
terms 
yields a purely 
non-deterministic extension of $\lambda$-calculus, which recalls the one of De Liguoro and Piperno 
\cite{deLiguoroP95}. Allowing only multisets of linear terms gives the linear fragment of $\Lamr$, used by Ehrhard and Regnier 
to give a quantitative account to $\lambda$-calculus $\beta$-reduction through Taylor expansion~\cite{bohmtaylor,difftaylor}.


But to be effectively used, $\Lamr$ needs a clear operational semantics. 
In this paper we investigate the notion of standardization in it. Let us recall that a calculus has the 
standardization property when
every reduction sequence can be rearranged according to a predefined order between redexes. 
Namely a reduction is standard with respect to a given order if
at every reduction step the reduced redex is not a residual of a redex which, in the given order, precedes a previously reduced one.
In the case of $\lambda$-calculus, the standardization is based on the left-to-right order of redexes.

In $\Lamr$, as the elements of a multiset are not ordered, a notion of standardization would be based on a partial order between
redexes. A first result, corresponding to a weak notion of standardization, has been proved by Pagani and Tranquilli \cite{PaganiTranquilli09}, stating that the reductions of redexes inside reusable resources
can always be postponed. 
We define a stronger partial order between redexes, and we prove that the non-deterministic reduction enjoys the standardization property
with respect to it. Even though this order is not total, it is in fact
undefined if and only if the two redexes live in different elements of
a same multiset, so that any finer order would not be well-defined.
This result allows us to complete the characterization of may-solvability, defined in  \cite{paganironchi10}.
Let us stress that solvability is
a key notion for evaluation, since it identifies the meaningful programs, and a clear notion of output result of a computation. 
Since this calculus is non-deterministic, 
two different notions of solvability arise, one optimistic (angelical, may) and 
one pessimistic (demoniac, must). In particular, in \cite{paganironchi10,paganironchiFI10} a characterization of the may-solvability
has been given, from a syntactical and logical point of view. Here we provide an operational characterization, through
an abstract reduction machine, performing the non-deterministic reduction. The soundness and completeness of the machine
with respect to the notion of may-solvability comes from the standardization property.
 
Moreover we prove that the parallel reduction does not enjoy the same standardization property.
Namely we show that in this case any order between linear redexes cannot be sound.
This negative result is interesting, since it gives evidence to the deep difference between linear and non-deterministic reduction.

%
%


\section{Syntax}
\begin{figure}[t]
  \footnotesize
  \centering
\begin{tabular}{ch{15pt}|h{15pt}c}
\subfigure[Grammar of terms, bags, sums, expressions.]{\label{fig:grammar}%
    \begin{tabular}{nL<{{}}LbLrn}
    \Lamr:       &M,N, L,O         &\bnf x \mid \lambda x.M \mid MP        & terms\\[5pt]
    \Lambda^{(\oc)}: &\argu M,\argu N&\bnf M \mid \bang{M}                     & resources\\[5pt]
    \Lambda^b:     &P, Q, R        &\bnf 1\mid\app{\bag{\argu M}}P
    & bags\\[5pt]
      \Expr:         &A,B            &\bnf M \mid P                            & expressions\\[5pt]
    \Nat\langle\Lamr\rangle:         &\Sum M, \Sum N, \Sum L            &\bnf 0 \mid M\mid \Sum M + \Sum N                            & sums of terms\\[5pt]
    \Nat\langle\Lambda^b\rangle:         &\Sum P, \Sum Q, \Sum R            &\bnf 0 \mid P\mid \Sum P+\Sum Q                            &sums of bags\\[5pt] 
     \multicolumn{3}{nLn}{\Sum A, \Sum B\in\Nat\langle\Expr\rangle\ass\Nat\langle\Lamr\rangle\cup\Nat\langle\Lambda^b\rangle} & sums of expressions\\[5pt]
   \end{tabular}%
}
&
\subfigure[Notation on $\Nat\langle\Expr\rangle$.]{
\label{fig:notsum}
$
\begin{aligned}
\lambda x.(\sum_i M_i)&\ass\sum_i\lambda x. M_i\\
(\sum_iM_i)(\sum_jP_j)&\ass\sum_{ij}M_iP_{j}\\
\bag{(\sum_iM_i)}\sdot (\sum_jP_j)&\ass \sum_{ij}\bag{M_i}\sdot P_{j}\\
\bag{(\sum^{k}_iM_i\bang)}\sdot (\sum_jP_j)&\ass \sum_{j}\bag{\bang M_1,\dots,\bang M_k}\sdot P_{j}
\end{aligned}\textstyle
$
}
\end{tabular}
\caption{\footnotesize Syntax of the resource calculus.}
\label{fig:statics}
\end{figure}
\paragraph{The syntax of $\Lamr$.} Basically, we have three syntactical sorts: terms, that are in functional position, bags, 
that are in argument position and represent multisets of resources, and finite formal sums, that 
represent the possible results of a computation. Precisely, Figure~\ref{fig:grammar} gives 
the grammar for generating the set $\Lamr$ of \textdef{terms} and the set $\Lambda^b$ of \textdef{bags} 
(which are in fact finite multisets of \textdef{resources} $\Lambda^{(\oc)}$) together with their typical metavariables. 
A resource can be linear (it must be used exactly once) or not (it can be used ad libitum, also zero times), in the last case it is written with a 
$\oc$ superscript. 
Bags are multisets presented in multiplicative notation, so that $\app PQ$ is the multiset union, and $1=\bag\,$ is the empty bag: 
that means, $\app P1=P$ and $\app PQ=\app QP$.
It must be noted though that we will never omit the dot $\cdot$, to avoid confusion with application.
\textdef{Sums} are multisets in additive notation, with $0$ referring to the empty multiset, so that:
$\Sum M+0=\Sum M$ and $\Sum M+\Sum N=\Sum N+\Sum M$.
We use two different notations for multisets in order to underline the different role of bags and sums. 

An \textdef{expression} (whose set is denoted by $\Expr$) is either a term or a bag. Though in practice only sums of 
terms are needed, for the sake of the proofs we also introduce sums of bags and of expressions. 
The symbol $\Nat$ denotes the set of natural numbers, and
$\Nat\langle\Lamr\rangle$ (resp.\ $\Nat\langle\Lambda^{b}\rangle$) denotes the set of finite formal sums of terms 
(resp.\ bags).

The grammar for terms and bags does not include sums in any point, so that
in a sense they may arise only as a top level constructor. However, as an inductive
notation (and \emph{not} in the actual syntax) we extend all the constructors to
sums as shown in Figure~\ref{fig:notsum}. In fact, all constructors but the
$\bang{(\cdot)}$ are, as expected, linear in the algebraic sense, \emph{\emph{i.e.}}\ they commute with sums.
In particular, we have that $0$ is always absorbing but for the $\bang{(\cdot)}$ constructor, 
in which case we have $\bag{0^\oc}=1$. 
We refer to~\cite{intudiffnet,phdtranquilli} for the mathematical intuitions underlying the resource calculus.

We adopt $\alpha$-equivalence and all the usual $\lambda$-calculus conventions as per \cite{barendregt}. 



The pair reusable/linear has a counterpart in the following two different notions of substitutions: 
their definition, hence that of reduction, heavily uses the notation of Figure~\ref{fig:notsum}. 
%
\begin{figure}[t]
\footnotesize
\centering
$
  y\linear{N}{x}\ass 
  \begin{cases}
    N&\textrm{if $y=x$,}\\[5pt]
    0&\textrm{otherwise,}
  \end{cases}
\quad
\begin{aligned}
  (\lambda y. M)\linear{N}{x}&\ass \lambda y.M\linear{N}{x},
\\[2pt]
  (MP)\linear{N}{x}&\ass M\linear{N}{x} P + MP\linear{N}{x},
\end{aligned}
\quad
\begin{aligned}
  1\linear{N}{x}&\ass 0,
\\[2pt]
  (\bag{M}\cdot P)\linear{N}{x}&\ass\bag{M\linear{N}{x}}\cdot
  P+\bag{M}\cdot P\linear Nx,
\\[2pt]
  (\bag{M^\oc}\cdot P)\linear{N}{x}&\ass\bag{M\linear{N}{x},M^\oc}\cdot
  P+\bag{M^\oc}\cdot P\linear Nx.
\end{aligned}
$%
\caption{\footnotesize Linear substitution. In the abstraction case we suppose $y\notin\mathrm{FV}(N)\cup\{x\}$.}\label{fig:linsub}
\end{figure}
\begin{definition}[Substitutions]\label{def:subst}
We define the following substitution operations.
\begin{enumerate}[(i)]
\item $A\expo Nx$ is the usual $\lambda$-calculus (\emph{i.e.}\ capture free) substitution of $N$ for $x$. It is extended to sums as in $\Sum A\expo {\Sum N} x$ by linearity in $\Sum A$.
The form $A\expo{x+N}x$ is called \textdef{partial substitution}.

\item $A\linear Nx$ is the \textdef{linear substitution} defined inductively in Figure~\ref{fig:linsub}.
It is extended to $\Sum A\linear{\Sum N} x$ by bilinearity in both $\Sum A$ and $\Sum N$.

\item $A\gen{N^{(\oc)}}{x}$, defined by
$A\gen{N}{x}\ass A\linear Nx$ and $A\gen{\bang N}{x}\ass A\expo {N+x}x$, is
the \textdef{resource substitution}, and moreover
$A\gen{B}{x}$, defined by
$A\gen{[N_1^{(!)},\dots,N_n^{(!)}]}x=A\gen{N_1^{(!)}}{x}\cdots\gen{N_n^{(!)}}{x}$ (assuming $x \notin FV(B)$)
is  the \textdef{bag substitution}.
\end{enumerate}
\end{definition}

Roughly speaking, the linear substitution corresponds to the replacement of the resource to exactly one \emph{linear} 
occurrence of the variable. In the presence of multiple occurrences, all the possible choices are made, and the result
is the sum of them. For example 
$(y\bag{x}\bag{x})\linear{N}x = y\bag{N}\bag{x}+y\bag{x}\bag{N}$. In the case there 
are no free linear occurrences, then linear substitution returns $0$, morally an error message. 
For example $(\lambda y.y)\linear Nx=\lambda y.(y\linear Nx)=\lambda y.0=0$. Finally, in case of 
reusable occurrences of the variable, linear substitution acts on a linear copy of the variable, 
\emph{e.g.}\  $\bag{\bang x}\linear Nx=\bag{N,\bang x}$. 
\paragraph{The reductions of $\Lamr$.} 
A term context $C\hole{\cdot}$ (or a bag context $P\hole\cdot$) is defined by
extending the syntax of terms and bags by a distinguished free variable called 
\textdef{hole} and denoted by $\hole\cdot$.

Notice that in contexts the order of holes cannot be truly
established as bags are independent of order. So filling\footnote{We recall that hole substitution allows for variable capture.} the $k$ holes of a contexts by terms needs 
a bijective mapping $a$ from
$\setof{1,\dots,k}$ to hole occurrences in $C\hole\cdot$, and $C_{a}\hole{\vec M_i}$
denotes the replacement of the holes by $M_{1},...,M_{k}$ guided by this map. We can write
also $C\hole{\vec M_i}$, by considering an implicit map.

A (term, bag) context is \textdef{simple} if it contains exactly one occurrence of the hole. 
In this case we will write simply $C\hole{M}$ for the result of filling
of the hole with $M$.
%
%
A simple context is \textdef{linear} if its hole is not under the scope of a $\bang{(\,)}$ operator, and it is
\textdef{applicative} if it has the hole not in a bag.
As usual the (simple/applicative/linear) context closure of a relation $R$ is
the one relating $C\hole t$ and $C\hole {t'}$ when $t\mathrel R t'$ and $C$ is
of the appropriate kind.

We define 
two kinds of reduction rule, called parallel and non-deterministic. Moreover the parallel
reduction can be further divided into
baby-step and giant-step, 
the former being a decomposition of the latter. Baby-step is more atomic, performing one substitution 
at a time, while the giant-step is closer to $\lambda$-calculus $\beta$-reduction, wholly consuming its redex in one shot. 
\begin{definition}[\cite{intudiffnet,phdtranquilli}]\label{def:giantbaby}
\begin{itemize}
\item[(i)] The \textdef{parallel} reductions are defined as follows:
\begin{itemize}
\item The \textdef{baby-step} reduction $\redto[b]$ is defined by the simple context closure of the following relation (assuming $x$ 
not free in $N$):
\begin{gather*}
(\lambda x. M)1\redto[b]M\expo{0}{x}\quad(\lambda x. M)\app{\bag{N}}{P}\redto[b]
(\lambda x. M\linear{N}{x})P\\
(\lambda x. M)\app{\bag{\bang{N}}}{P}\redto[b](\lambda x. M\expo{N+x}{x})P
\end{gather*}
\item The \textdef{giant-step} reduction $\redto[g]$ is defined by the simple context closure of the following relation: 
$$ (\lambda x.M)P \redto[g] M\gen{P}{x}\expo{0}{x}$$
\end{itemize}
\item[(ii)] The \textdef{non-deterministic} reduction is the relation
$M \redto[nd] N$ if and only if $M \redto[g] N + \mathbb{A}$ for some $\mathbb{A}$.

  \end{itemize}
\end{definition}

\begin{notation} For any reduction
  $\redto[\epsilon]$ (the ones listed above and the ones to come), we denote by
  $\redto[\epsilon*]$ its reflexive-transitive closure.
  $\rho: M \redto[\epsilon *] N$ denotes a particular reduction sequence from $M$ to $N$,
  and $|\rho|$ its length.
  \end{notation}
%
\paragraph{$\Lamr$ and $\lambda$-calculus.}
In $\lambda$-calculus, arguments can be used as many times we want, so it is easy to inject it in $\Lamr$ through the
following translation $(.)^{*}$:
$$
(x)^{*}=x, \;\ (\lambda x.M)^{*}= \lambda x. (M)^{*}, \;\ (MN)^{*}= (M)^{*}\bag{\bang{(N)^{*}}} 
$$
On terms of $\Lamr$ which are translations of $\lambda$-terms, the giant reduction becomes the usual $\beta$-reduction.


\section{Standardization}
In this section we will prove that the non-deterministic reduction enjoys a standardization property.
As we recalled already in the introduction, the standardization property is based on an order relation between redexes.
We can define it formally as follows:
\begin{definition}
Let $\prec$ be an order on positions in terms (which is extended to an order on
subterms of a given term). Suppose $\rho$ is a reduction chain, and let
$M_i$ and $R_i$ be the $i$-th term and fired redex in $\rho$ respectively.
We say that $\rho$ is \emph{$\prec$-standard} if for every $i$ we have that
$R_{i+1}$ is not the residual of a redex $R'$ in $M_i$ such that $R'\prec R_i$.
\end{definition}

We will prove that non deterministic reduction in $\Lamr$ enjoys the standardization property
with respect to the order $\prec_r$, which is the partial order on positions in $\Lambda^r$ terms that, intuitively,
gives precedence to linear positions over non-linear ones, and then orders
linear positions left-to-right, with the proviso that positions inside the same bag
be not comparable. The formal definition follows.

\begin{definition}[Linear left-to-right order]\label{def:order}
For two subterms $S_1$ and $S_2$ inside the expression $\Sum A$, we say that
$S_1\prec_r S_2$ in $\Sum A$ if and only if any of the following happens:
\begin{itemize}
 \item $S_2$ is a subterm of $S_1$;
 \item $S_1$ is linear in $\Sum A$ while $S_2$ is not;
 \item $S_1$ and $S_2$ are both linear in $\Sum A$, $\Sum A=MP$, $S_1$ is in $M$ and $S_2$ is in $P$.
 \item $S_1$ and $S_2$ are subterms of the same proper subexpression $\Sum B$ of $\Sum A$, and
 $S_1\prec_r S_2$ in $\Sum B$;
\end{itemize}
\end{definition}
\begin{example}
 $S_1\prec_r S_2$ in both $\lambda x.x\bag{\bang S_2}\bag{S_1}$ and
 $\lambda x.x\bag{S_1}\bag{S_2}$,
 while they are incomparable in $\lambda x.x\bag{S_1,S_2}$.
\end{example}

Our starting point is the division of redexes in two classes, outer and inner.

\begin{definition}[\cite{PaganiTranquilli09}]\label{def:outerreduction}
Let $\epsilon\in\{\red{b}, \red{g}, \red{nd}\}$. The \textdef{outer $\epsilon$-reduction} $\outbeta[\epsilon]$ is the 
 \emph{linear} context closure of the $\epsilon$-steps given in
Definitions~\ref{def:giantbaby}. 
A \textdef{non-outer $\epsilon$-reduction}, called \textdef{inner} is
defnoted by $\nonoutbeta[\epsilon]$.
\end{definition}
In other words, an outer reduction does not reduce inside reusable resources, so
an outer redex (\emph{i.e.}\ a redex for $\outbeta[\epsilon]$) is a redex not under the scope of a 
$\bang{(\cdot)}$ constructor. In particular a term corresponding to a 
$\lambda$-term has at most one outer-redex, which coincides with the head-redex.
Pagani and Tranquilli stated in some sense a weak standardization property for the giant reduction,
proving that inner redexes can always be postponed. Their result
can easily be extended to other reductions, in particular to the non-deterministic one.

\begin{theorem}[\cite{PaganiTranquilli09}]\label{the:patra}
Let $\epsilon\in\{\red{b}, \red{g}, \red{nd}\}$. $M \redto[\epsilon^*]\Sum A $ implies $M \redto[o\epsilon^{*}]\Sum A' $ and $\Sum A' \redto[i\epsilon^{*}]\Sum A$.
\end{theorem}

We introduce now a further classification between outer redexes.
\begin{definition}
The set of \emph{leftmost} redexes $\L(M)$ of a term $M$ or a bag $P$ are defined
  inductively by:
  $$\begin{aligned}
    \L(x) &\ass  \emptyset,\\ \L(\lambda x.M) &\ass\L(M)
  \end{aligned}
  \;\;
    \L(MP) \ass
    \begin{cases}
      \{MP\} & \text{if $M=\lambda x.M'$,}\\
      \L(M) & \text{otherwise, if $\L(M)\neq \emptyset$}\\
      \L(P) & \text{otherwise}
    \end{cases}
    \;\;
    \begin{aligned}
      \L(1) &\ass \emptyset,\\
      \L([\bang M]\cdot P &\ass \L(P),\\
      \L([M]\cdot P) &\ass \L(M) \cup \L(P)
    \end{aligned}
  $$
\end{definition}

In regular $\lambda$-calculus, the set $\L(M)$ is at most a singleton, and
$\prec_r$-standardness collapses to the regular notion of left-to-right order of redexes.

\begin{fact}
  Redexes in $\L(M)$ are exactly the $\prec_r$-minimal elements among
  all redexes of $M$.
\end{fact}

In the following, we will consider in particular the non-deterministic reduction. So, let us introduce some notation.

\begin{notation}
Let $M \redto[ndo] N$.
$M\redto[lm]N$ denotes that the reduction fires a
redex in $\L(M)$, while we write $M\xredto[\neg lm]N$ if the redex
is not a leftmost one. Moreover $M\redto[o]N$ and $M\redto[i]N$ will be short for 
for $M\redto[ndo]N$ and $M\redto[ndi]N$ respectively.

\end{notation}

\begin{lemma}\label{lem:rightpushing}
We have the following facts on non-leftmost reduction.
\begin{itemize}
\item $\rho:\lambda x.M \xredto[\neg lm*] N$ if and only if $N = \lambda x.M'$
  and $\rho':M \xredto[\neg lm*] M'$ with $|\rho| = |\rho'|$;
\item $\rho:MP \xredto[\neg lm*] N$ if and only if $N = M'P'$, $\rho':M
  \xredto[\neg lm*] M'$ and $\rho'':P \xredto[o*] P'$ with $|\rho| =
  |\rho'| + |\rho''|$;
\item $\rho: [M]\cdot P \xredto[\neg lm*] Q$  if and only if$Q = [M']\cdot P'$,
  $\rho':M \xredto[\neg lm*] M'$ and $\rho'':P \xredto[\neg lm*] P'$
  with $|\rho| = |\rho'| + |\rho''|$;
\item $\rho: [\bang M]\cdot P \xredto[\neg lm*] Q$  if and only if$Q = [\bang
  M]\cdot P'$ and $\rho'':P \xredto[\neg lm*] P'$ with $|\rho| =
  |\rho''|$.
\end{itemize}
\end{lemma}
The proof of standardization is based on an inversion property between outer redexes,
saying that a not-leftmost reduction followed by a leftmost one can always be replaced by a leftmost followed by 
an outer. This is the upcoming \autoref{lem:inversion}. In order to get it we first prove the following intermediate properties.

\begin{lemma} \label{lem:triangolo}
If $O \redto[o] O'$ then
$\forall L' \in O'\gen{Q}{x}\expo{0}{x}
\exists L \in O\gen{Q}{x}\expo{0}{x}$
such that $L \redto[o] L'$.
\end{lemma}
\begin{proof}
We will prove that $\forall L' \in O'\gen{Q}{x}
\exists L \in O\gen{Q}{x}$
such that $L \redto[o] L'$. Then the statement of the lemma follows easily.
By induction on $O$. 
\begin{enumcases}
\item $O=x$ and $O=y$ are not possible.
\item $O=\lambda y.M$.
By inductive hypothesis.
\item $O=(\lambda y.M)P$.
There are three cases: $\lambda y.M  \redto[o]  \lambda y.M'$, $ P \redto[o] P'$,
$(\lambda y.M)P \redto[o] O' \in M\gen{P}{y}\expo{0}{y}$.
Let $M  \redto[o]  M'$.
$(\lambda y.M)P \gen{Q}{x}= \sum_{Q_{1},Q_{2}}((\lambda y.M) \gen{Q_{1}}{x})(P\gen{Q_{2}}{x})$,where
$Q_{1},Q_{2}$ range over all the possible decomposition of $Q$ into two parts, counting the reusable resources
with all the possible multiplicities. This means that in case $Q_{1},Q_{2}$ are considered two different subterms also
in case they are syntactically equal.
By inductive hypothesis, for all $L' \in (\lambda y.M') \gen{Q_{1}}{x}$ there is
$L \in (\lambda y.M) \gen{Q_{1}}{x}$ such that $L \redto[o] L'$, and the result follows by transitivity of $\in$.
The case $ P \redto[o] P'$ is similar.

Let $(\lambda y.M)P \redto[o] O' \in M\gen{P}{y}\expo{0}{y}$.
Then we have that the substitution $(\lambda y.M)P\gen{Q}{x}$ is equal to the sum $\sum_{Q_{1},Q_{2}}(\lambda y.M \gen{Q_{1}}{x}) (P\gen{Q_{2}}{x}) $, where
$Q_{1},Q_{2}$ range as before. Since each component of this sum is a redex (the substitutions do not modify the
external shape of the terms), we can reduce each redex, so obtaining that for all
$L\in (\lambda y.M \gen{Q_{1}}{x}) (P\gen{Q_{2}}{x})$, $L \redto[o] L' \in M  \gen{Q_{1}}{x}
\gen{P\gen{Q_{2}}{x}}{y} \expo{0}{y}$.
On the other side, $M\gen{P}{y}\expo{0}{y}\gen{Q}{x}$ is equal to the sum
$\sum_{Q_{1},Q_{2}}M \gen{Q_{1}}{x}
\gen{P\gen{Q_{2}}{x}}{y} \expo{0}{y} $, and the proof is done.

\item $O=MP$ and $O'= M'P$ or $O=MP$ and $O'= MP'$.
All by inductive hypothesis.\qedhere
\end{enumcases}

\end{proof}

\begin{lemma} \label{lem:quadrato}
If $Q \redto[o]  Q'$ then
$\forall L' \in O\gen{Q'}{x} \expo{0}{x}
\exists L \in O\gen{Q}{x} \expo{0}{x}$
such that $L \redto[o] L'$.
\end{lemma}

\begin{proof}
By induction on $Q$. $Q$ cannot be $1$ as it would be normal.

If $Q=\lbrack H \rbrack \cdot P$ then
$O \gen{Q}{x} \expo{0}{x} =
O \linear{H}{x} \gen{P}{x} \expo{0}{x}$.
We proceed by cases:
\begin{enumcases}
\item The reduction is on $P$, \emph{i.e.}\ $ [H] \cdot P \redto[o] [H] \cdot P' $.
For all $N \in O \linear{H}{x}$, $N\gen{P}{x} \redto[o]N\gen{P'}{x}$. Then we have by induction that
for all $L \in N\gen{P}{x} \expo{0}{x}$ there is $L' \in N \gen{P'}{x} \expo{0}{x}$ such that $L \redto[o] L'$.
So the result follows.

\item The reduction is on $H$, \emph{i.e.}\ $ [H] \cdot P \redto[o] [H'] \cdot P $).
Let us set $
O\linear{H}{x}\gen{P}{x}\expo{0}{x} =
(O_1 + ... + O_k) \gen{P}{x}\expo{0}{x}$,
where $H$ occurs in all $O_{j}$ ($1 \leq j \leq k$), since the substitution is linear.
Let $O_{j} \redto[og] O^{j}_{1}+...+O^{j}_{m}$ by reducing the occurrence of $H$ in it.
So $O_{j} \redto[o] O^{j}_{i}$ ($1 \leq i \leq m_{j}$), and, by \autoref{lem:triangolo},
for all $L' \in O^{j}_{i} \gen{P}{x} \expo{0}{x}$, there is $L\in O_{j}\gen{P}{x} \expo{0}{x}$
such that $L\redto[o] L' $. Since $O_{j}\in O\linear{H}{x}$ and $O^{j}_{i}\in O\linear{H'}{x}$, the proof follows.\
\end{enumcases}

If $Q= \lbrack H^! \rbrack \cdot P$
the reduction on $P$, and the case is similar to the first case of the previous point.\qedhere



\end{proof}

\begin{lemma}[Inversion] \label{lem:inversion}
$M \redto[\neg lm] M'  \redto[lm] N $ implies $M  \redto[lm] M'' \redto[ndo] N$, for some $M''$.
\end{lemma}
\begin{proof}
We proceed by induction on $M$.\\
Let $M = \lambda \vec{y} .(\lambda x.O)Q P_1...P_j...P_n$, so ${\cal L}(M)=\{(\lambda x.O)Q\}$.
Non leftmost reductions on $M$ can be done in $O$, in $Q$ or in $P_j$ ($1 \leq j \leq n$). We procede by cases:
\begin{enumcases}
\item The reduction is on $P_j (1 \leq j \leq n)$.].
Let $P_j \redto[g] P'_j + \mathbb{S}$.
We have that
$$M \redto[\neg lm] \lambda \vec{y}. (\lambda x.O) Q P_1 ... P'_j ... P_n \redto[g]
\lambda \vec{y}.O \gen{Q}{x}\expo{0}{x} P_1 ... P'_j ... P_n.$$
Moreover, by reducing the leftmost redex
$$M \redto[g] \lambda \vec{y}.O \gen{Q}{x} \expo{0}{x} P_1 ... P_j ... P_n =
\lambda \vec{y}.(O_1 + ... +O_k) P_1 ... P_j ... P_n,
$$
so that
$$M \redto[lm] \lambda \vec{y}.O_h P_1 ... P_j ... P_n \redto[o]
\lambda \vec{y}.O_h P_1 ... P'_j ... P_n$$
for all $1\leq h \leq k$.

\item The reduction is on $Q$.
Let $Q \redto[o] Q'$ and let $M \redto[\neg lm] \lambda \vec{y}.(\lambda x.O)Q' P_1 ... P_n  \redto[lm]
\lambda \vec{y}. \overline{M} P_1 ... P_n$, where $\overline M$ is such that
$\overline{M} \in O\gen{Q'}{x} \expo{0}{x} $.
Moreover, by reducing the leftmost redex, we also have the reduction
$M \redto[g] \lambda \vec{y}. O\gen{Q}{x} \expo{0}{x} P_1 ... P_n $.
By \autoref{lem:quadrato}, $Q \redto[o]Q'$ implies that for all $L'\in O\gen{Q'}{x} \expo{0}{x}$,
there exists $
L \in O\gen{Q}{x} \expo{0}{x}$ such that $L \redto[o] L'$.
So there is $\overline{\overline{M}}\in O\gen{Q}{x}\expo{0}{x}$ such that
$M \redto[lm] \lambda \vec{y}. \overline{\overline{M}}  P_1 ... P_n \redto[o] \lambda \vec{y}. \overline{M} P_1 ... P_n.$

\item The reduction is in O.
Let $O \redto[o] O' $, and let
$M \redto[\neg lm] \lambda \vec{y}.(\lambda x.O')Q P_1...P_n \redto[lm] \lambda \vec{y}. \overline{M}P_1...P_n$,
where $\overline{M}$ is such that $\overline M \in O'\gen{Q}{x} \expo{0}{x}$.
Again if we reduce the leftmost redex, we have the reduction
$M \redto[g] \lambda \vec{y}. O \gen{Q}{x} \expo{0}{x} P_1 ... P_n$.
$O \redto[o]O'$ implies, by \autoref{lem:triangolo}, $\forall L' \in O'\gen{Q}{x} \expo{0}{x},
 \exists L \in O \gen{Q}{x} \expo{0}{x}$ such that
$L \redto[o] L'$. So there is $\overline{\overline{M}}\in O \gen{Q}{x} \expo{0}{x}$ such that we
can compose the reductions
$M \redto[lm] \lambda \vec{y}. \overline{\overline{M}}  P_1 ... P_n \redto[o] \lambda \vec{y}.\overline{M}
 P_1 ... P_n $.
\end{enumcases}

Let $M=\lambda \vec{y}.x P_1 ... P_n$, and let $P_{i} \redto[\neg lm]P'_{i}$ and $P_{j} \redto[lm]P'_{j}$.
In case $i \not = j$, the proof is trivial. In case $i=j$ the proof is by induction on $P_{i}$.\qedhere

\end{proof}
\begin{corollary}\label{cor:reordering}
  If $\rho:M\xredto[o*]M'$ then there are $\sigma:M\xredto[lm*]M''$
  and $\pi:M''\xredto[\neg lm*]M'$ with $|\sigma| +
  |\pi| = |\rho|$.
\end{corollary}

\begin{lemma}\label{lem:divide}\
\begin{enumerate}[(i)]
 \item Given $\rho : M \redto[lm*] N$ and $\sigma: N \xredto[o*] L$,
then $\rho\sigma: M\xredto[o*] L$ is
 $\prec_r$-standard if and only if $\sigma$ is.
 In particular every chain of leftmost reductions is $\prec_r$-standard.
 \item Given $\rho : M\xredto[o*] N$ and $\sigma: N \redto[i*] L$,
then $\rho\sigma: M \redto[nd*] L$ is
 $\prec_r$-standard if and only if both $\rho$ and $\sigma$ are.
\end{enumerate}
\end{lemma}
 \begin{proof}
 The result follows easily from the definition of $\prec_r$.
 \end{proof}
Now we can prove that the non-deterministic outer reduction is $\prec_{r}$-standard.
\begin{lemma}[Non-deterministic outer standard reduction]
  \label{lem:stdouter}
  If $M \xredto[o*] N$, then there is a $\prec_r$-standard
  non-deterministic outer reduction from $M$ to $N$.
\end{lemma}
\begin{proof}
   We reason by induction on the pair $(p,s)$, where $p=|\rho|$ is the
  length of the reduction sequence $\rho:M \xredto[o*] N$, and $s$ is
  the number of symbols in $M$.  By \autoref{cor:reordering},
  there is a reduction $\sigma_l:M \redto[lm*] M'$ and $\sigma_r: M'
  \xredto[\neg lm*] N$ with $|\sigma_l| + |\sigma_r| = |\rho| = p$.
  If $|\sigma_l|>0$ then inductive hypothesis applies to $\sigma_r$,
  giving $\prec_r$-standard $\sigma'_r: M' \xredto[o*] N$, which
  gives that $\sigma_l\sigma'_r: M\xredto[o*] N$ is
  $\prec_r$-standard by \autoref{lem:divide}.
  In case $\sigma_r: M \xredto[\neg lm*] N$ is the
  whole reduction, the proof is by cases on $M$.
 The only non-obvious case is when $M=LP$: by \autoref{lem:rightpushing} we have $N=L'P'$ and
    $\rho': L \xredto[\neg lm*] L'$ and $\rho'': P \xredto[o*]
    P'$. We can apply inductive hypothesis to both as $|\rho'| +
    |\rho''| = |\rho_r|$, and get $LP \xredto[\neg lm*] L'P
    \xredto[o*] L'P'$. Now assuming that this is not
    $\prec_r$-standard leads to a contradiction to the definition at
    the seam, since all linear positions in $L''$ are $\prec_r$ with respect to
    those in $P$.
\end{proof}
In order to prove that also inner reductions can be standardized, we need to introduce the notion of 
{\em outer shape} of a term.
\begin{definition}
  The \emph{outer shape} $\ell(M)\hole\cdot$ of a term $M$ is a context that
  is $M$ with holes replacing all exponential arguments of $M$'s
  bags.

  Formally, extending the definition to bags, we define
  $\ell(\,.\,)\hole\cdot$ inductively as follows.
  
  $\begin{array}{lll}
    \ell(x)\hole\cdot = x, &
    \ell(\lambda x.M) \hole\cdot= \lambda x.\ell(M)\hole\cdot,&
    \ell(MP)\hole\cdot = \ell(M)\hole\cdot\ell(P)\hole\cdot,\\
    \ell(1)\hole\cdot  = 1,&
    \ell(\bag M\cdot P)\hole\cdot = \bag{\ell(M)\hole\cdot}\cdot\ell(P)\hole\cdot,&
    \ell(\bag{\bang M}\cdot P)\hole\cdot= \bag{\bang{\hole\cdot }}\ell(P)\hole\cdot.
  \end{array}$
\end{definition}
\begin{property}\label{lem:inner_outer}\
\begin{enumerate}[(i)]
\item \label{lem:inner_shape} $M\xredto[i*]N$ if and only if $\ell(M)\hole\cdot=\ell(N)\hole\cdot$, and there are $k$
  terms $M'_i$ and $k$ terms $N'_i$ such that
  $M=\ell(M)_a\hole{\vec M'_i}$, $N=\ell(M)_a\hole{\vec N'_i}$ and
  $M'_i\xredto[nd*]N'_i$ for each $i$.
\item \label{lem:outer_shape_std}
  If $M=\ell(M)_{a}\hole{\vec M'_i}$ and
  $\rho_i:M'_i\xredto[nd*]M''_i$
  are standard, then there is a standard 
  $\rho':M\xredto[i*]\ell(M)_{a}\hole{\vec M''_i}$.
\end{enumerate}
\end{property}
\begin{proof}\
\begin{itemize}
\item[i)] The if direction is a direct consequence of how $\red{i}$ is
  defined and of context closedness of the reduction. We thus move to
  the only if direction.

  First, let us show that the property to prove is preserved by
  composition of reduction chains.

  Suppose
  $M\xredto[i*]N\xredto[i*]O$ with $M=\ell(M)[\vec M'_i]$,
  $N=\ell(M)_{a_1}[\vec N'_i]=\ell(M)_{a_2}[\vec N''_i]$
  and $O=\ell(M)_{a_3}[\vec O'_i]$. We can suppose
  $a_1=a_2$ by re-indexing (namely using
  $\ell(N)_{a_1}[\vec N''_{a_2^{-1}(a_1(i)}]$ and
  $\ell(O)_{a_3'}[\vec N''_{a_2^{-1}(a_1(i)}]$ with
  $a_3'=a_3\circ a_2^{-1}\circ a_1$). So we just
  forget the bijections employed, and then we have by hypothesis
  $M'_i\xredto[nd*]N'_i=N''_i\xredto[nd*]O'_i$, which is what is
  needed.

  Now, we can prove the property by reducing to the case of a single
  inner reduction, as composing multiple ones of them preserves the
  property.

  Take $M\redto[i]N$: the result follows by a straightforward
  induction on how the reduction is defined.

\item[ii)] The idea is that the reductions in the subterms can be freely
  rearrenged.

  Let us reason by generalizing to expressions and by structural
  induction on $\Sum A$.
  \begin{enumcases}
    \item $\Sum A=x$ or $\Sum A=1$: nothing to prove.
    \item $\Sum A=\lam x.N$: straightforward application of inductive
      hypothesis.
    \item $\Sum A=NP$, with $\ell(A)\hole\cdot=\ell(N)\hole\cdot\ell(P)\hole\cdot$: we can partition
      $M'_i$ into what goes in $\ell(N)\hole\cdot$ and what goes in
      $\ell(P)\hole\cdot$. We can suppose that
      $\Sum A=(\ell(N)[M'_1,\dots,M'_h])(\ell(P)[M'_{h+1},\dots,M'_k])$
      without loss of generality,
      and by inductive hypothesis get standard $\sigma:N\xredto[i*]N'$ and
      $\rho:P\xredto[i*]P'$ (with $N'$ and $P'$ the correct pluggings of
      $\ell(N)$ and $\ell(P)$).

      Now, if we reduce $\Sum A=NP\xredto[i*]N'P\xredto[i*]N'P'$
      following first $\sigma$ and then $\rho$, the resulting
      reduction must be standard as all positions in $P$ are greater
      than those in $N$ according to $\prec_r$.
    \item $\Sum A=[N]\cdot P$: exactly as above, but without any constraint
      on the order in which the reductions are composed.
    \item $\Sum A=[\bang N]\cdot P$, with $\ell(\Sum A)=\bag{\bang{\hole\cdot}}\cdot
      \ell(P)\hole\cdot$: suppose that $M'_1=N$ and
      $P=\ell(P)[M'_2,\cdots,M'_k]$. By inductive hypothesis we have a
      standard $\rho:P\xredto[i*]P'=\ell(P)[M''_i]_{i=2}^k$, and as
      positions in $[\bang N]$ and \emph{non-linear} positions in $P$
      are incomparable, we can freely combine the reductions on $M'_1$
      and $P$ to get a standard one.\qedhere
  \end{enumcases}
 \end{itemize}
\end{proof}
Now we are able to show the desired result.
\begin{theorem}[Standardization]\label{thm:standard-inner}
If $M\redto[nd*]M'$, then there is a $\prec_r$-standard chain from $M$ to $M'$.
\end{theorem}
\begin{proof}
  By structural induction on $M'$, the term where the reduction ends.
     First, applying \autoref{the:patra}, we get
  $\sigma:M\xredto[o*]M''$ and $\rho:M''\xredto[ndi*]M'$. Now we
  strive to obtain two standard chains $\sigma':M\xredto[o*]M''$ and
  $\rho':M''\xredto[ndi*]M'$ to obtain the chain $\sigma'\rho'$ which
  is standard by \autoref{lem:divide}. The existence of a standard
  $\sigma'$ is assured directly by \autoref{lem:stdouter}, so we
  need to concentrate on finding $\rho'$.
  By using \hyperref[lem:inner_outer]{\autoref*{lem:inner_outer}(\ref*{lem:inner_shape})}, we get
  $M''=\ell(M')\hole{N_1,\dots,N_k}$, $M'=\ell(M')\hole{N'_1,\dots,N'_k}$ and
  $\rho_i:N_i\xredto[nd*]N'_i$. As all $N'_i$ are structurally
  strictly smaller than $M'$, we can apply inductive hypothesis on
  each $\rho_i$ and get standard $\rho_i':N_i\xredto[nd*]N'_i$. Then
  using \hyperref[lem:outer_shape_std]{\autoref*{lem:inner_outer}(\ref*{lem:outer_shape_std})} we can glue back those
  reductions into the standard reduction $\rho':M''\xredto[i*]M'$.
\end{proof}
\begin{example}
Let $I=\lambda x.x$, $M_{1}=I \bag{\bang{ ((\lambda xy.x)\bag{\bang{I}}\bag{\bang{I}})}}$,
$M_{2}= I \bag{\bang{I}}$, and
let $M=\lambda x.x \bag{\bang{M_{1}},\bang{M_{2}}}$.
The following reduction is standard:
$M_{1}=I \bag{ \bang{ ((\lambda xy.x)\bag {\bang{I}}\bag{\bang{I}})   } } \redto[lm]
(\lambda xy.x)\bag{\bang{I}}\bag{\bang{I}} \redto[lm] (\lambda y.I)\bag{\bang{I}} \redto[lm] I.$
As $M_{2} \redto[lm] I$, the following is standard too
\begin{multline*}
\lambda x.x \bag{  \bang   {    (I \bag{\bang{ ((\lambda xy.x)\bag{\bang{I}}\bag{\bang{I}})}}     )},\bang{(I \bag{\bang{I}})}    } \redto[i]
\lambda x.x \bag{  \bang   {   ((\lambda xy.x)\bag{  \bang{I} }\bag{\bang{I} })   },\bang{(I \bag{\bang{I}})}    } \redto[i]
\\
\lambda x.x \bag{   \bang{((\lambda xy.x)\bag{\bang{I}}\bag{\bang{I}})   },
\bang{I}} \redto[i]
\lambda x.x \bag{\bang{((\lambda y.I)\bag{\bang{I}})},\bang{I}} \redto[i] 
\lambda x.x \bag{\bang{I},\bang{I}}.\end{multline*}
%
\end{example}

Let us notice that, as opposed to the weak form of standardization given in \autoref{the:patra}, the
$\prec_{r}$-standardization does not hold for parallel reduction. A counterexample is the following.
\begin{example}
Let $I_{0}$ and $I_{1}$ denote two occurrences of the identity $\lambda x.x$, and let
$M=I_0[I_1[x^!,y^!]]\redto[g] I_0[x]+I_0[y]\redto[g] x + I_0[y]$ by reducing the inner redex first.
But reducing the leftmost redex first we obtain
$M \redto[g] I_1[x^!,y^!] \redto[g] x+y$.
So the previous result cannot be obtained by a standard reduction. 
 \end{example}



\section{Solvability Machine}
The standardization result proved in the previous section allows us to design an abstract reduction machine characterizing the may-solvable terms in $\Lamr$.
A term of $\lambda$-calculus is solvable whenever there is a outer-context reducing it to the identity \cite{barendregt}. In the resource 
calculus, terms appear in formal sums,
so (at least) two different notions of solvability arise, related to a may and must operational semantics, respectively.
We will treat the former only.
%
\begin{definition}\label{def:maysolvability}
A simple term $M$ is \textdef{\maysol{}} whenever there is a linear applicative--context $C\hole\cdot$ such that $C\hole M\xredto[nd\ast]I$.
\end{definition}
May-solvability has been  completely characterized from both a syntactical and logical point of view in \cite{paganironchi10}. Syntactically, a term $M$ is \maysol\ if and only if it is may-outer normalizable.
An expression is an \textdef{outer normal form} ($onf$) if it has no redex but under the scope of a $\bang{(\,)}$, and
consequently
a term $M$ is \textdef{may-}\textdef{outer normalizable} if and only if 
$M \redto[nd\ast] N$, where $N$ is a $onf$ ($N$ is called a $monf$ of $M$). 
Logically, a particular intersection type assignment system has been defined,
typing all and only the \maysol\  terms.

We now will complete the job, characterizing may solvability from an operational point of view.
The following property is obvious.

\begin{property}\label{leftmost-empty}
$M$ is in $onf$ if and only if $\L(M)=\emptyset$.
\end{property}

The abstract reduction machine (called $ND$-machine)
proves statements of the shape $M \NDMachine N$, where $M,N$ are simple terms and
$N$ is a $onf$. 
The $ND$-machine uses an auxiliary machine,
the $B$-machine, performing the reductions on bags. The two machines are shown in Figure~\ref{fig: non-det-machine}.  

\begin{figure}
\centering
\subfigure[The ND reduction machine.]{\label{fig:NDmachine}
$\begin{gathered}
\infer[(\lambda)]{\lambda x.M \NDMachine \lambda x. M'}{M \NDMachine M'}\qquad
\infer[(end)]{M \NDMachine M}{M \mbox{ is in onf }}\qquad 
\infer[(head)]{xP_{1}...P_{m} \NDMachine xP'_{1}...P'_{m}}{P_i \BagMachine P'_{i} \quad (1 \leq i \leq m)}
\\
\infer[(0)]{(\lambda x.M)1 P_{1}...P_{m} \NDMachine M'}{M \{0/x\} P_{1}...P_{m} \NDMachine M'}\qquad
\infer[(\beta)]{(\lambda x.M)[N]\cdot P P_{1}...P_{m} \NDMachine M''}
{M\left<N/x\right> = M' + \Sum A \quad (\lambda x.M')P P_{1}...P_{m}\NDMachine M'' }
\\
\infer[(! \beta)]{(\lambda x.M)[N^{!}] \cdot P P_{1}...P_{m} \NDMachine M''}{M\left\{N+x/x\right\} = M' + \Sum A  \quad
  (\lambda x.M')P P_{1}...P_{m} \NDMachine M''}\\[5pt]
\end{gathered}$
}

\hrulefill

\subfigure[The auxiliary $B$ machine]{\label{fig:Bmachine}
$$
\infer[(1_{b})]{1 \BagMachine 1}{}
\qquad
\infer[(b)]{[M] \cdot P \BagMachine [N] \cdot P'}{M \NDMachine N \quad P \BagMachine P'}
\qquad
\infer[(!b)]{[M^{!}] \cdot P \BagMachine [M^{!}] \cdot P'}{P \BagMachine P'}
$$
}
\label{fig: non-det-machine}
\end{figure}


Some comments are in order. First of all, the machine performs the baby outer reduction, on a leftmost redex. 
Rules $(\lambda)$, $(end)$ and $(0)$ are self-explanatory. Rule $(head)$ implements the definition of 
$monf$; note that in this rule the order 
in which
the arguments are reduced does not matter.
Non-determinism appears in rules $(\beta)$ and $(! \beta)$. Indeed, if the result of the substitution is a sum, one of its addends  is randomly chosen. 
The auxiliary machine $B$ performs the reductions on bags. Note that the rule ($!b$) implements the notion of outer-reduction.
Remember that $0$ is not  a term, so it can be neither an input nor an output of the machine. So in rules $(0)$, $(\beta)$ 
and $(! \beta)$ the machine transition is undefined if the result of the substitution is $0$.
We will write $M\DNDMachine$ to denote that for any run of the machine on $M$ either it does not stop or it is undefined.

\begin{example}
$(\lambda zy.y)\bag{x} \DNDMachine$. In fact, trying to apply rule $\beta$, the machine needs to compute
$(\lambda y.y)\left<x/z\right>$, which is equal to $0$, so the premises of the rule are not satisfied.\\
$(\lambda x. x \bag{\bang{x}})(\lambda x. x \bag{\bang{x}})\DNDMachine$. In fact, the machine on this input does not stop. 
Notice that this term corresponds to an unsolvable term in the $\lambda$-calculus. \\
Let $F=\lambda xy.y$.
Then $(\lambda x. y \bag{x} \bag{x})\bag{F, I} \red[lm]$ reduces non deterministically to $y\bag{F}\bag{I} + y\bag{I} \bag{F}$. It is easy to check that there are two 
machine computations such that in one
$(\lambda x. y \bag{x} \bag{x})\bag{F, I}\NDMachine y\bag{F}\bag{I}$ while in the other
$(\lambda x. y \bag{x} \bag{x})\bag{F, I}\NDMachine y\bag{I} \bag{F}$.\\
$(\lambda x.y\bag{\bang{x}})\bag{\bang{I}, \bang{F}} \redto[g] y\bag{\bang{I}, \bang{F}}$, by reducing the leftmost redex. 
The unique machine computation for this input gives $(\lambda x.y\bag{\bang{x}})\bag{\bang{I}, \bang{F}}\NDMachine y\bag{\bang{I}, \bang{F}}$.
\end{example} 


\begin{theorem}\label{theo:soundCompl}
\begin{enumerate}[(i)]\
\item (Soundness) \label{theo:soundness}
If $M \NDMachine N$ then $M  \redto[lm*]  N$, and $N$ is a $\hnf$.
\item (Completeness)\label{theo:completeness}
Let $M$ be may-outer-normalizable and let $N$ be a  $\mhnf$ of $M$. 
There is a machine's computation proving $M \NDMachine N'$, 
where $N'$ is a $\mhnf$ of $M$ and $N' \xredto[\neg lm*] N$.
\end{enumerate}
\end{theorem}
\begin{proof}[Proof (sketch)]
Point~\eqref{theo:soundness} is proved by mutual induction on the rules of the two machines.
Point~\eqref{theo:completeness} is an immediate consequence of the $\prec_{r}$-standardization property.
\end{proof}

%
%

\paragraph*{Acknowledgements.}
We would like to thank Michele Pagani for his interesting and useful suggestions.

\bibliographystyle{eptcs}
\newcommand{\online}[1]{Available at \url{#1}}

\end{document}